\documentclass[prb,preprint]{revtex4-1}
\usepackage{epsfig} 
\usepackage{amsmath,amssymb}
\usepackage{graphicx}
\usepackage{caption}
\usepackage{subcaption}
\usepackage[usenames,dvipsnames]{xcolor}

\newcommand{\be}{\begin{equation}}
\newcommand{\ee}{\end{equation}} 
\newcommand{\bq}{\begin{eqnarray}} 
\newcommand{\eq}{\end{eqnarray}}
\newcommand{\f}{\frac}
\newcommand{\E}{{\cal E}}
\newcommand{\D}{{\cal D}}

\newcommand{\proofend}{\raisebox{1.3mm}{\fbox{\begin{minipage}[b][0cm][b]{0cm}
\end{minipage}}}}
\newenvironment{proof}{\noindent{\it Proof:} }{\mbox{}\hfill \proofend\\\mbox{}}
\newtheorem{lemma}{Lemma}

\begin{document}

\title{On equivalent resistance of electrical circuits}

\author{Mikhail Kagan} 
\email{mak411@psu.edu}
\affiliation{Division of Science and Engineering, The Pennsylvania State University, 
Abington College, Abington, PA 19116, USA}

\date{\today}

\begin{abstract}
While the standard (introductory physics) way of computing the equvalent resistance of {\em non-trivial} electrical ciruits is based on Kirchhoff\rq{}s rules, there is a mathematically and conceptually simpler approach, called the method of {\em nodal potentials}, whose basic variables are the values of electric potential at the circuit\rq{}s nodes. 
In this paper, we review the method of nodal potentials
and illustrate it using the Wheatstone bridge as an example. At the end, we derive---in a closed form---the equivalent resistance of a generic circuit, which we apply to a few sample circuits. The final result unveils a curious interplay between electrical circuits, matrix algebra, and graph theory and its applications to computer science. 
The paper is written at a level accessible by undergraduate students 
 who are familiar with matrix arithmetic. For the more inquisitive reader, additional proofs and technical details are provided in the appendix. 
\end{abstract}
\maketitle

\section{Introduction} 


The first thorough mathematical description of electrical circuits goes back to Gustav Kirchhoff. \cite{Kirchhoff} Ever since, this topic has given rise to a great number of pedagogical articles. On one hand, circuits are a source of elegant problems and solutions \cite{Aitchison} or interesting experiments. \cite{Denardo, Mead, Gordon, Falcon} On the other hand, electrical circuits can be used to visually and intuitively illustrate some complicated (or often misunderstood by students \cite{Rosenthal, Brown}) electrodynamical concepts. \cite{Jackson} 
Some authors apply powerful methods known from elsewhere to compute the equivalent resistance of nontrivial electrical circuits. \cite{Jeng, Cserti, Woong_Kook} In this paper, we primarily focus on computing the equivalent resistance of a generic circuit. Remarkably, unlike many of the works cited above, the method we employ here doesn't require the individual resistances 
 to have specific values, e.g., to be all equal, or to form a symmetric or lattice-like structure. At the same time, this method is simple enough to be accessible to students familiar with the basics of matrix arithmetic.

From the mathematical perspective, a resistive electric circuit can be understood as a graph whose nodes are connected via links (edges) 
with assigned numerical values of resistance ($R$). In addition, two nodes are assumed to be ``connected to the battery,'' which fixes the \textit{potential difference} (voltage, $\E$) between the two nodes. The battery gives rise to currents flowing along the edges of the circuit. These edge currents ($I$) are related to the potential differences across the corresponding nodes ($\Delta V$) via Ohm\rq{}s law:
\be\label{Ohm}
I=\f{\Delta V}{R}.
\ee
The total current flowing through all the edges coming out of a battery node is called the output current ($I_{\rm out}$). The total current that flows into the other battery node is equal to the output current. One of the important characteristics of an electrical circuit is its {\em equivalent resistance} ($R_{\rm eq}$). It is defined as the resistance of a single resistor (edge) that, if it were to replace the whole circuit, would result in the same amount of the output current. In other words,
\be\label{R_eq}
R_{\rm eq}=\f{\E}{I_{\rm out}}.
\ee

Although the output current is explicitly present in the definition of equivalent resistance, it is clear (e.g., on dimensional grounds) that $R_{\rm eq}$ depends only on the given edge resistances. In some simple cases, it can be computed without finding the output current. For instance, if a node is only shared by two edges, such edges are said to be connected \textit{in series}, and the individual resistances are merely added. 
If two resistors are on edges connected across the same pair of nodes 
such a connection is called \textit{in parallel}, and the equivalent resistance is computed as the reciprocal of the sum of reciprocals. Both of these situations can be generalized straightforwardly for more than two resistors:
		\bq
		R_{\rm eq} = &R_1 + R_2+\ldots, \quad\quad\quad &\mbox {for connection in series;} \\
		R_{\rm eq} =&\left(R_1^{-1} + R_2^{-1}+\ldots\right)^{-1}, \quad &\mbox {for connection in parallel.} 
		\eq
Some larger circuits allow for reduction by identifying sub-circuits whose elements are connected either in series or in parallel. Replacing such sub-circuits by single equivalent resistances turn by turn may result in a trivial circuit, thus giving an algorithm for computing equivalent resistance. 

At the same time, it is clear that not every circuit can be simplified this way. For example, a circuit with no parallel edges and whose nodes are at least three-valent (such nodes are referred to as \textit{junctions}) cannot be reduced in the above sense. The simplest non-simplifiable circuit is depicted in Fig.~\ref{Fig:Wheatstone}.
\begin{figure}[h!] 
\centerline{\includegraphics[width=6cm, keepaspectratio]{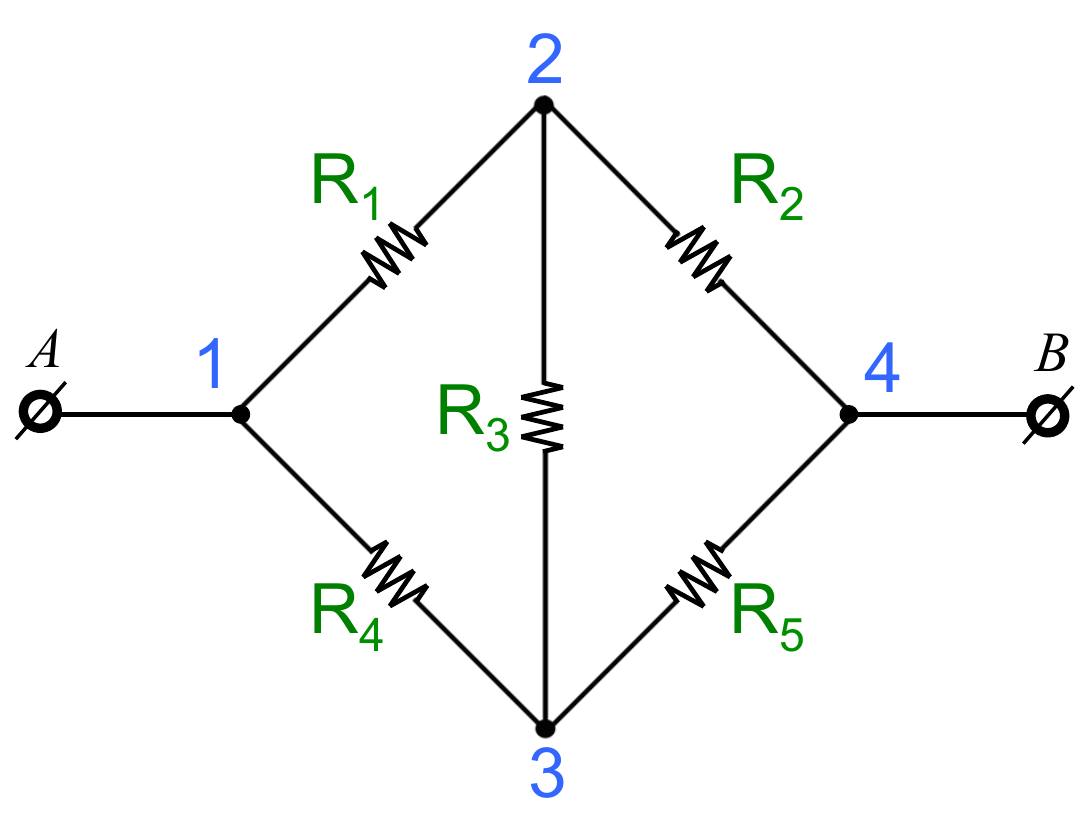}} \caption{The simplest non-simplifiable circuit (Wheatstone bridge). There are no connections in series or in parallel.\label{Fig:Wheatstone}} 
\end{figure}
It is easy to see that there are no elements connected in series or parallel. In order to determine the equivalent resistance of such a circuit, one typically introduces the unknown edge currents and writes down Kirchhoff\rq{}s rules. After solving those equations for the unknown currents, one can compute the output current and obtain the equivalent resistance from Eq.~(\ref{R_eq}). 


There is also an alternative approach to finding unknown quantities associated with electrical circuits, called the {\em method of nodal potentials} (see, for example, Refs.~\onlinecite{Reitz, DeCarlo}; there are also variational alternatives to Kirchhoff's loop rule presented in Ref.~\onlinecite{Variation}). In this method the basic variables are the potentials at the circuit's nodes rather than the edge currents. The advantage of this method is that it usually involves significantly fewer variables than in Kirchhoff's framework. In addition, Kirchhoff's rules come in two types (junction and loop) which have distinct algebraic structure. On the contrary, the equations for the unknown nodal potentials are all of one type (following from the junction rule).  
This presents a major simplification for the analytical description of electrical circuits, as well as for finding various quantities of interest.

This paper is organized as follows. We start by reviewing the method of nodal potentials for electrical circuits, highlighting some immediate implications of the formalism and working out a sample circuit (the Wheatstone bridge). For generic values of the edge resistances, the latter is somewhat nontrivial when using Kirchhoff's rules, \cite{Students} but is quite straightforward using the method of nodal potentials (see Eq.~(\ref{CeqWheatstone}) for the equivalent conductance). In Section \ref{Sec:SigmaEq}, we then proceed to derive a closed formula for the equivalent resistance/conductance of an arbitrary resistive circuit (see Eq.~(\ref{Ceq})). Related results exist in the graph theoretical literature,\cite{Wu, R-eff} expressing {\em resistance distance} of a graph with 1-$\Omega$ resistors on each edge. Ref. \onlinecite{R-eff} 
utilizes the Moore-Penrose (pseudo) matrix inverse and does not lead to easily derived conclusions. 
In this paper, the derivations, as well as the final result, are quite intuitive and accessible to undergraduate students familiar with basic matrix arithmetic. In the last section, we discuss generalizations of the formula, as well as its possible relation to matrix algebra and graph theory and its applications to computer science.


\section{Review of the method of nodal potentials}

In this section we review a convenient description of resistive DC-circuits containing one battery and investigate its implications. We shall see how the method of nodal potentials can simplify the analysis of an electric circuit compared to the standard Kirchhoff's rules technique. The same approach can be applied to DC-circuits with several batteries, as well as to generic AC-circuits. Such generalizations are discussed in the last section.

\subsection{Formalism and notation}\label{Sec:Formalism}
Consider a circuit containing $n$ nodes, such that nodes $1$ and $n$ are connected to the positive and negative terminals of the battery respectively (Fig. \ref{Fig:GenericCircuit}). 

\begin{figure}[h!] 
\centerline{\includegraphics[width=6cm, keepaspectratio]{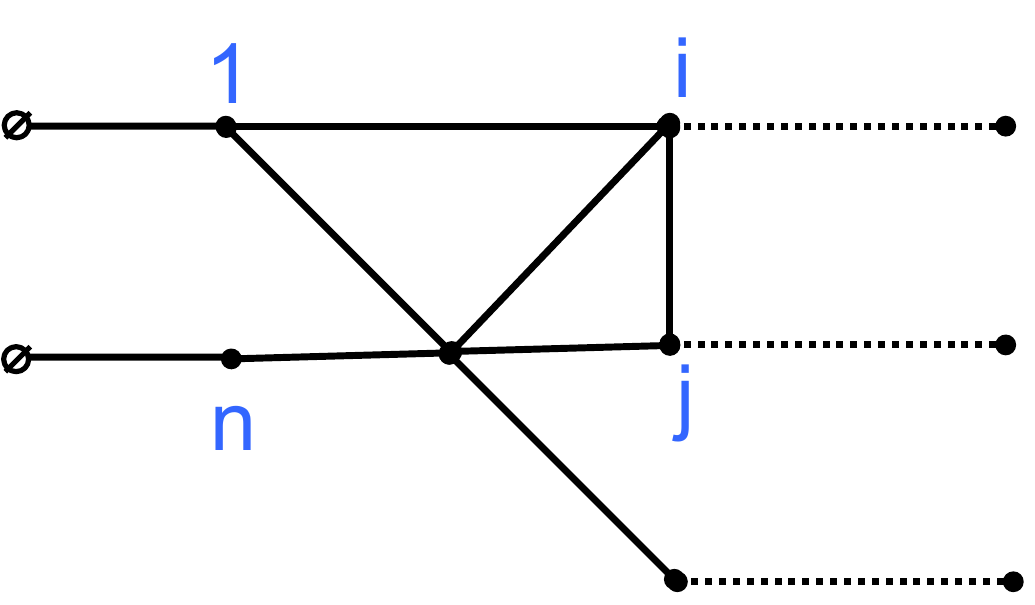}} \caption{A generic circuit with $n$ nodes. The positive (higher potential) battery terminal is connected to node $1$ and delivers the output current to the circuit ($I_{\rm out}$). The current flows back into the other terminal of the battery via node $n$. \label{Fig:GenericCircuit}} 
\end{figure}
Every link connecting two nodes $i$ and $j$ is assumed to have known resistance $R_{ij}\equiv R_{ji}$. In fact from now on, it will be more convenient to use the conductance rather than resistance, defined as 
\be
\sigma_{ij}:=\f{1}{R_{ij}},
\ee
thus giving rise to the conductance matrix $\sigma_{ij}$. It is easy to see that the rules for computing equivalent conductance are reversed compared to those for equivalent resistance: 
\bq
\sigma_{\rm eq} = &\sigma_A + \sigma_B+\ldots, \quad\quad\quad &\mbox {for connection in parallel} \label{sigmaPar}
\\
\sigma_{\rm eq} = &\left(\sigma_A^{-1} + \sigma_B^{-1}+\ldots\right)^{-1}, \quad &\mbox {for connection in series} \label{sigmaSer}
\eq


Without loss of generality, we can make the following assumptions:
\begin{itemize} 
\item $V_1=\E$ and $V_n=0$. Since electric potential is defined up to a constant, we fix one of them to zero. The potential at the positive terminal, $\E$, can be used as a unit.
\item  Every node is connected to every other node. If, in reality, some nodes $i$, $j$ etc. do not share a link, we simply put $\sigma_{i,j}=0$.   
\end{itemize} 
The latter assumption will allow us to not worry about the circuit's topology and concentrate on purely algebraic description. In fact, if we only focus on non-simplifiable circuits, every node, except possibly the $1^{\rm st}$ and $n^{\rm th}$ ones, will have at least three edges. Indeed, a node with only two edges would imply a connection of the edges in series which could be replaced by a single edge. 
Note that each column/row of the conductance matrix  for a non-simplifiable circuit must have at least three non-zero entries. 

We now define the edge current $I_{ij}$ as the current flowing from node $i$ to node $j$. Its expression can be readily written in terms of the nodal potentials as
\be\label{CurrentDef}
I_{ij}:=\sigma_{ij}\left( V_i-V_j\right)\equiv -I_{ji},
\ee
which is manifestly anti-symmetric. In particular, the ``diagonal'' current $I_{ii}=0$. Since electric current flows from a higher potential to a lower one, this definition sets an outgoing current to be positive, whereas an incoming current would be negative.

We conclude this subsection with the following remark. A minimally connected, non-simplifiable circuit has at least $\sim 3n/2$ edges, hence the same number of unknown currents. At the same time, the number of unknown nodal potentials is $(n-2)$, which can be substantially less than the number of currents. Thus the method of nodal potentials deals with fewer variables at the on-set.
\subsection{Implications}
Assuming that the variables describe a real circuit which has specific unique values of the nodal potentials, Kirchhoff's loop rules will be automatically satisfied by construction: any closed loop will come back to the same value of potential, hence making the overall potential difference zero. 

For a generic circuit, Kirchhoff's junction rules take the form
\begin{equation}\label{JunctionRules}
\sum\limits_{j=1}^{n} I_{ij} \equiv \sum\limits_{j=1}^{n} \left(V_i-V_j\right)\sigma_{ij}=0, \quad {\rm for }\, \, i\neq 1 \,\, {\rm or} \,\,n.
\end{equation}
Here, each sum constitutes the total nodal current for $i=2, ..., n-1$. Nodes $1$ and $n$ are not included, as the battery provides non-zero nodal currents at those nodes.  In fact, Eqs. (\ref{JunctionRules}) imply the following 
\begin{lemma}\label{lemma1}
The total current flowing into the $n^{\rm th}$ node is equal to the total current flowing out of the $1^{\rm st}$ node.
\end{lemma}
\begin{proof}
Since the edge currents are anti-symmetric, $I_{ij}=-I_{ji}$, the sum over both indices $\sum \limits_{i,j=1}^n I_{ij}=0$. Splitting the summation over $i$ into $i=1$, $i=n$ and the rest yields
\be
\sum \limits_{j=1}^n I_{1,j}+\sum \limits_{j=1}^n I_{n,j}+\sum \limits_{i=2}^{n-1}\left(\sum\limits_{j=1}^{n} I_{ij}\right)=0.
\ee 
By virtue of Eq.(\ref{JunctionRules}), each summand in the last term vanishes. The first and second terms are the output and (negative) input currents respectively, which proves the lemma.
\end{proof}

While the statement of the lemma was not {\rm mathematically} obvious from the construction, it makes clear physical sense: since there is no accumulation of charge in the circuit, the total current coming out of one terminal of the battery has to equal the current flowing back into its other terminal. 

There is another statement which is obvious from the physical point of view, but non-trivial mathematically: \cite{Topology} 
\begin{lemma}\label{lemma2}
The values of the nodal potentials in a resistive circuit connected to a single battery should lie (strictly) between the lower battery voltage and the higher one, i.e., between $0$ and $\E$. In other words, the battery sets the lowest and the highest possible potential in the circuit.
\end{lemma}
\begin{proof}
Suppose that the maximum potential is attained at node $m\neq 1$. Then all the neighboring nodes would have a lower potential resulting in currents from node $m$
\be
I_{mj}=\sigma_{mj}\left(V_m-V_j\right)>0
\ee
to be outgoing. The latter would violate the junction rule (\ref{JunctionRules}). Thus the maximum potential can't be attained at any node other than the $1^{\rm st}$ one. Similarly, the minimum potential is attained at node $n$. Hence all the nodal potential values lie between $V_1$ and $V_n$, which proves the lemma.
\end{proof}

In the next session we illustrate how the method of nodal potentials helps to determine the equivalent conductance of the Wheatstone bridge circuit with arbitrary edge conductances.
\subsection{An example: Wheatstone bridge circuit}\label{Sec:Wheatstone}
We now assume that the values of edge resistance in Fig.\ref{Fig:Wheatstone} are given and find the equivalent resistance (conductance) of the circuit. Notice that if one was to solve this problem using Kirchhoff's rules, one would have to deal with six unknown currents: one in each resistor plus the output current. At the same time, there are only two unknown nodal potentials: $V_2$ and $V_3$, since $V_1=\E$ and $V_4=0$. For convenience, we also label the edge conductances similarly to the original resistances:  
\be
\sigma_{12}\equiv \sigma_1 = \frac{1}{R_1}, \quad \sigma_{24}\equiv \sigma_2 = \frac{1}{R_2}, \quad \sigma_{23}\equiv \sigma_3 = \frac{1}{R_3}, \quad \sigma_{13}\equiv \sigma_4 = \frac{1}{R_4}, \quad \sigma_{34}\equiv \sigma_5 = \frac{1}{R_5}.
\ee

It is easy to see that there are exactly two junction equations for this circuit, one for node 2 and one for node 3:

\bq
I_{21}+I_{24}+I_{23}=\sigma_1\left(V_2-V_1\right)+\sigma_2\left(V_2-V_4\right)+\sigma_3\left(V_2-V_3\right)=0 \label{V2}\\
I_{31}+I_{34}+I_{32}=\sigma_4\left(V_3-V_1\right)+\sigma_5\left(V_3-V_4\right)+\sigma_3\left(V_3-V_2\right)=0\label{V3}
\eq

Setting $V_1=\E$ and $V_4=0$ and solving the equations on the righthand side for the unknown potentials $V_2$ and $V_3$, we obtain:

\be\label{V2V3}
V_2 =\E \frac{\sigma_1 c_3+\sigma_3\sigma_4}{c_2c_3-\sigma_3^2}, \quad V_3 =\E \frac{\sigma_4 c_2+\sigma_1\sigma_3}{c_2c_3-\sigma_3^2},
\ee
where $c_2\equiv \sigma_1+\sigma_2+\sigma_3$ and $c_3\equiv \sigma_3+\sigma_4+\sigma_5$. 

It is important to understand that for a connected circuit ($0<\sigma_{ij}<\infty$) the denominator in (\ref{V2V3}) can never be zero. To make it more transparent, we rewrite the denominator as 
\be
\D={\sigma_{3}(\sigma_1+\sigma_2+\sigma_{4}+\sigma_5)+(\sigma_1+\sigma_2)(\sigma_4+\sigma_5)}.
\ee
Since each term in $\D$ is non-negative, it can only be zero if each term is zero. Irrespective of whether $\sigma_3=0$, the latter implies that $\sigma_1=\sigma_2=\sigma_4=\sigma_5=0$. This corresponds to a disconnected circuit, such that nodes 2 and 3 are completely isolated, which makes their potentials undetermined. \cite{Connectivity}

In order to find the equivalent conductance, it is easiest to consider the current flowing into node 4, $I_{24}+I_{34}=\sigma_2\left(V_2-V_4\right)+\sigma_5\left(V_3-V_4\right)$. Setting $V_4=0$ and using the nodal potentials in (\ref{V2V3}), we obtain
\be\label{CeqWheatstone}
\sigma_{\rm eq}=
\frac{\sigma_2 V_2+\sigma_5 V_3}{\E}=\frac{\sigma_1 \sigma_2c_3+\sigma_2\sigma_3\sigma_4+\sigma_1\sigma_3\sigma_5+\sigma_4\sigma_5  c_2}{\sigma_{3}(\sigma_1+\sigma_2+\sigma_{4}+\sigma_5)+(\sigma_1+\sigma_2)(\sigma_4+\sigma_5)}.
\ee
Looking at the answer, we see that it is a ratio of two polynomials of degree $(n-1)$ and $(n-2)$ respectively. This clearly guarantees the correct units of conductance. Moreover, each polynomial is a sum of non-negative terms. As explained in Ref. \cite{Connectivity} , this means that (\ref{V2V3}) and (\ref{CeqWheatstone}) are neither zero nor infinity for any connected circuit. 

Finally, the Wheatstone bridge has a well known feature that for a special arrangement of edge resistances (conductances), there is no current through the middle wire (labeled with $\sigma_3$). We can arrive this condition by setting equal potentials at the ends of the wire $V_2=V_3$. This yields  $\sigma_1 \sigma_5= \sigma_2 \sigma_4$ or the standard
\begin{equation}\label{WheatSym}
\f{R_1}{R_2}=\f{R_4}{R_5}.
\end{equation}
In the next section we shall generalize the expression for the nodal potentials (\ref{V2V3}) and derive the equivalent conductance for an arbitrary circuit.

\section{Method of nodal potentials in a generic circuit}
We now revisit the generic circuit displayed in Fig. \ref{Fig:GenericCircuit} with the same assumptions as in Sec. \ref{Sec:Formalism}: the battery terminals read $V_1=\E$ and $V_n=0$ and all the edge conductances $\sigma_{ij}$ are given.
\subsection{Expressions for the nodal potentials}
 We can write the junction equations, analogous to (\ref{V2}) and (\ref{V3}), for all the nodal potentials, including $V_1$ and $V_n$. Collecting similar terms, we cast those equations in the following matrix form
\be\label{SigmaMatrix}
\begin{pmatrix}
c_1 & -\sigma_{12} & -\sigma_{13} & \dots & -\sigma_{1, n-1} &-\sigma_{1,n}\\
-\sigma_{21} & c_{2} & -\sigma_{23} & \dots & -\sigma_{2, n-1} &-\sigma_{2,n}\\
-\sigma_{31} & -\sigma_{32} & c_{3} & \dots & -\sigma_{3, n-1} &-\sigma_{3,n}\\
\vdots & \vdots & \vdots &\ddots &\vdots & \vdots\\ 
-\sigma_{n-1,1} & -\sigma_{n-1,2} & -\sigma_{n-1,3} & \dots & c_{n-1} &-\sigma_{n-1,n} \\
-\sigma_{n,1} & -\sigma_{n,2} & -\sigma_{n,3} & \dots & -\sigma_{n,n-1} & c_n 
\end{pmatrix}
\begin{pmatrix}
{V_1} \\
{\bf V_2}\\
{\bf V_3}\\
\vdots\\
{\bf V_{n-1}}\\
{V_n}
\end{pmatrix}
=
\begin{pmatrix}
{\bf I_{\rm out}} \\
0\\
0\\
\vdots\\
0\\
{\bf I_{\rm in}}
\end{pmatrix},
\ee
where the unknown quantities are highlighted in boldface, and the diagonal elements $c_{i}=\sum\limits_{j=1}^n \sigma_{ij}$. We denote this matrix $\Sigma$, \cite{SigmaMatrix} and will also need $\Sigma^\prime$, its upper-left sub-matrix $(n-1)\times (n-1)$
\be
\Sigma^\prime=\begin{pmatrix}
c_1 & -\sigma_{12} & -\sigma_{13} & \dots & -\sigma_{1, n-1} \\
-\sigma_{21} & c_{2} & -\sigma_{23} & \dots & -\sigma_{2, n-1} \\
-\sigma_{31} & -\sigma_{32} & c_{3} & \dots & -\sigma_{3, n-1} \\
\vdots & \vdots & \vdots &\ddots &\vdots \\ 
-\sigma_{n-1,1} & -\sigma_{n-1,2} & -\sigma_{n-1,3} & \dots & c_{n-1} \\
\end{pmatrix} \quad {\rm and} \quad
\label{SigmaPP}
\Sigma^{\prime\prime}=\begin{pmatrix}

 c_{2} & -\sigma_{23} & \dots & -\sigma_{2, n-1} \\
 -\sigma_{32} & c_{3} & \dots & -\sigma_{3, n-1} \\
 \vdots & \vdots &\ddots &\vdots \\ 
 -\sigma_{n-1,2} & -\sigma_{n-1,3} & \dots & c_{n-1} \\
\end{pmatrix},
\ee 
the lower-right sub-matrix of $\Sigma^\prime$ of size $(n-2)\times(n-2)$. Note that the sum of all the elements of $\Sigma$ is zero, which implies that the $n$ equations in (\ref{SigmaMatrix}) are not independent. This is clearly consistent with the fact that $I_{\rm in}=-I_{\rm out}$. Hence from now on we skip the $n^{\rm th}$ equation. 

In order to solve for the unknown nodal potentials $V_2$, $V_3$, $...$, $V_{n-1}$, we rearrange the equations in (\ref{SigmaMatrix}) as follows. The first element of each row of $\Sigma$ multiplies $V_1=\E$. We carry this term to the righhand side of each equation in (\ref{SigmaMatrix}). In the first row, we also carry the unknown output current $I_{\rm out}$ to the lefthand side. Now these equations take the form 
\be\label{SigmaRear}
\begin{pmatrix}
-1 & -\sigma_{12} & -\sigma_{13} & \dots & -\sigma_{1, n-1} \\
0 & c_{2} & -\sigma_{23} & \dots & -\sigma_{2, n-1} \\
0 & -\sigma_{32} & c_{3} & \dots & -\sigma_{3, n-1} \\
\vdots & \vdots & \vdots &\ddots &\vdots \\ 
0 & -\sigma_{n-1, 2} & -\sigma_{n-1, 3} & \dots & c_{n-1} 
\end{pmatrix}
\begin{pmatrix}
I_{\rm out} \\
V_2\\
V_3\\
\vdots\\
V_{n-1}
\end{pmatrix}=
\E\begin{pmatrix}
-c_1\\
\sigma_{21}\\
\sigma_{31}\\
\vdots\\
\sigma_{n-1, 1}
\end{pmatrix}.
\ee
Applying Cramer's rule to rows $2$ through $(n-1)$, we obtain the following expressions for the nodal potentials
\be\label{NPexpressions}
V_k=\E\frac{\det \Sigma^{\prime\prime}_k}{\det \Sigma^{\prime\prime}}.
\ee
Here $\Sigma^{\prime\prime}_k$ is the matrix obtained from $ \Sigma^{\prime\prime}$ by substituting its $k^{\rm th}$ column by $(\sigma_{21}, \sigma_{31}, ..., \sigma_{n-1,1})^{\rm T}$. On physical grounds, for any connected circuit the determinant in the denominator should be non-zero. This, however, is not  so obvious from the mathematical point of view. See Appendix \ref{App:Determinant} for more detail.

\subsection{Derivation of the equivalent conductance of a generic circuit}\label{Sec:SigmaEq}

In principle, one can find the equivalent conductance of the circuit in Fig.\ref{Fig:GenericCircuit} using the same method as in Sec. \ref{Sec:Wheatstone}. Specifically, we could use the nodal potentials found in (\ref{NPexpressions}) to compute the output current, using (\ref{JunctionRules}) for $n=1$, and substitute in
\be
\sigma_{\rm eq}=\frac{I_{\rm out}}{\E}.
\ee
There is, however, a more economical way to arrive at the expression for the equivalent conductance. Interestingly, it can be obtained in a closed form. Moreover the expression of $\sigma_{\rm eq}$ in terms of the individual conductances (which are assumed given) is universal and does not require prior finding nodal potentials. Applying Cramer's rule to the first row of (\ref{SigmaRear}), we can find $I_{\rm out}$
\be
I_{\rm out}=\E\frac{
\begin{vmatrix}
-c_1 & -\sigma_{12} & -\sigma_{13} & \dots & -\sigma_{1, n-1} \\
\sigma_{21} & c_{2} & -\sigma_{23} & \dots & -\sigma_{2, n-1} \\
\sigma_{31} & -\sigma_{32} & c_{3} & \dots & -\sigma_{3, n-1} \\
\vdots & \vdots & \vdots &\ddots &\vdots \\ 
\sigma_{n-1, 1} & -\sigma_{n-1, 2} & -\sigma_{n-1, 3} & \dots & c_{n-1} 
\end{vmatrix}}
{\begin{vmatrix}
-1 & -\sigma_{12} & -\sigma_{13} & \dots & -\sigma_{1, n-1} \\
0 & c_{2} & -\sigma_{23} & \dots & -\sigma_{2, n-1} \\
0 & -\sigma_{32} & c_{3} & \dots & -\sigma_{3, n-1} \\
\vdots & \vdots & \vdots &\ddots &\vdots \\ 
0 & -\sigma_{n-1, 2} & -\sigma_{n-1, 3} & \dots & c_{n-1} 
\end{vmatrix}}=\E\frac{-\det \Sigma^\prime}{-\det \Sigma^{\prime\prime}}.
\ee
Therefore, the equivalent conductance reads
\be\label{Ceq}
\sigma_{\rm eq}=\frac{\det \Sigma^\prime}{\det \Sigma^{\prime\prime}}.
\ee
As before, the answer is a ratio of two polynomials of degree $(n-1)$ and $(n-2)$, which clearly has correct units. Each determinant is of the form similar to that of Sec. \ref{Sec:Wheatstone}. They both are non-zero (positive) for any connected circuit. See Appendix \ref{App:Determinant} for more detail. Eq.(\ref{Ceq}) constitutes the main result of this paper. In the following section we consider a few specific electrical circuits and use (\ref{Ceq}) to compute their equivalent conductance. We discuss general  properties  of (\ref{Ceq}) in the last section.

\section{Examples}
Even though the application of (\ref{Ceq}) is straightforward for an arbitrary electrical circuit, it is still interesting to consider a few examples and make sure that the formula work correctly. In this section, we analyze relatively simple circuits. More general properties of (\ref{Ceq}) are discussed in the Appendices. 
\subsection{Wheatstone bridge}
Now, as we have the closed formula (\ref{Ceq}), we can compute the equivalent conductance of the Wheatstone bridge directly. Using the conductance labels of Sec. \ref{Sec:Wheatstone}, the corresponding matrix is 
\be
\Sigma=\begin{pmatrix}
c_1 & -\sigma_{1} & -\sigma_{4} & 0  \\
-\sigma_{1} & c_{2} & -\sigma_{3} & -\sigma_{2}  \\
-\sigma_{4} & -\sigma_{3} & c_{3} & -\sigma_{5}  \\
0&-\sigma_{2} & -\sigma_{5} & c_{4}  
\end{pmatrix},
\ee
where $c_1=\sigma_1+\sigma_4$, $c_2=\sigma_1+\sigma_2+\sigma_3$, $c_3=\sigma_3+\sigma_4+\sigma_5$, and $c_4=\sigma_2+\sigma_5$. Applying (\ref{Ceq}), we obtain
\be
\sigma_{\rm eq}=\frac{
\begin{vmatrix}
c_1 & -\sigma_{1} & -\sigma_{4}  \\
-\sigma_{1} & c_{2} & -\sigma_{3}   \\
-\sigma_{4} & -\sigma_{3} & c_{3}     
\end{vmatrix}}
{\begin{vmatrix}
 c_{2} & -\sigma_{3}   \\
 -\sigma_{3} & c_{3}     
\end{vmatrix}
}=\frac{\sigma_1 \sigma_2 c_{3}+\sigma_2\sigma_3\sigma_4+\sigma_1\sigma_3\sigma_5+\sigma_4\sigma_5  c_{2}}{\sigma_{3}(\sigma_1+\sigma_2+\sigma_{4}+\sigma_5)+(\sigma_1+\sigma_2)(\sigma_4+\sigma_5)},
\ee
which agrees with (\ref{CeqWheatstone}).
\subsection{A three-node circuit}
\begin{figure}
        \centering
        \begin{subfigure}[b]{0.4\textwidth}
                \centering
                \frame{\includegraphics[width=\textwidth]{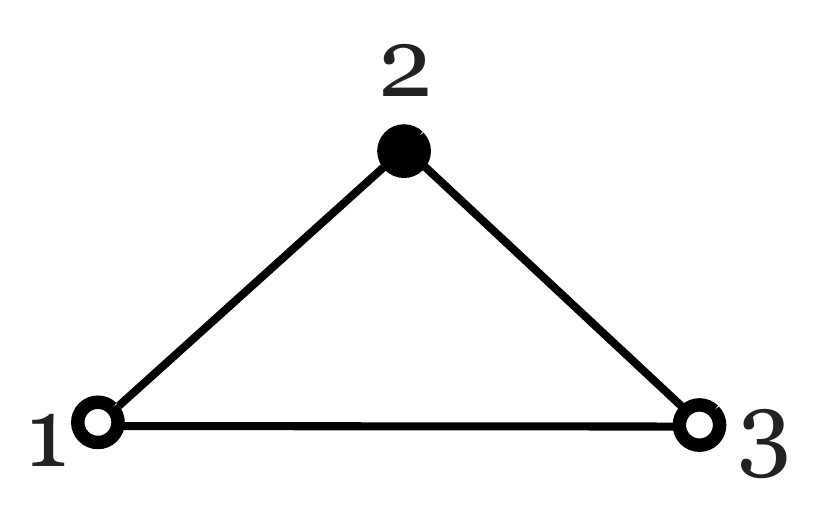}}
                \caption{A three-node circuit.}
                \label{Fig:Triangle}
        \end{subfigure}%
		\qquad
        ~ 
        \begin{subfigure}[b]{0.4\textwidth}
                \centering
                \frame{\includegraphics[width=\textwidth]{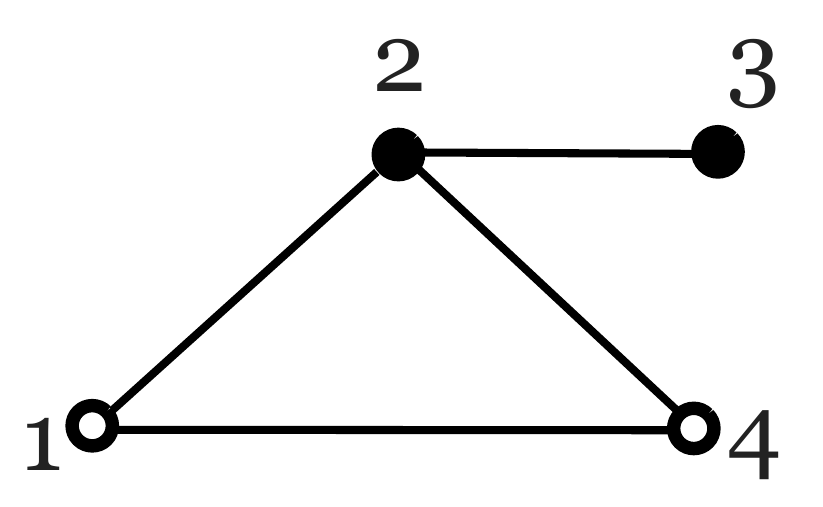}}
                \caption{Edge 2-3 doesn't carry current}
                \label{Fig:Triangle+1}
        \end{subfigure}
        \caption{Simple circuits. The first and last nodes are connected to the battery.}\label{Fig:Speed_Spin}
\end{figure}

In Appendix \ref{App:Series}, we prove that formula (\ref{Ceq}) yields the right result when one combines two resistors in series and rewrites the $\Sigma$-matrix accordingly.  For the three-node circuit in Fig.~\ref{Fig:Triangle}, however, we shall obtain the equivalent conductance directly. Before we do so, note that - based on Eqs. (\ref{sigmaSer}) and (\ref{sigmaPar}) - the equivalent conductance is 
\be\label{Triangle}
\sigma_{\rm eq} = \sigma_{13}+\frac{\sigma_{12} \sigma_{23}}{\sigma_{12}+ \sigma_{23}}
\ee
The $\Sigma$-matrix for this circuit is 
\be
\Sigma = \begin{vmatrix}
c_1 & -\sigma_{12} & -\sigma_{13}  \\
-\sigma_{21} & c_{2} & -\sigma_{23}   \\
-\sigma_{31} & -\sigma_{32} & c_{3}     
\end{vmatrix}
\ee
Setting the symmetric entries equal and using (\ref{Ceq}), we get
\be
\sigma_{\rm eq}=\frac{
\begin{vmatrix}
c_1 & -\sigma_{12} \\
-\sigma_{21} & c_{2} 
\end{vmatrix}
}{c_2}=\frac{\sigma_{12}\sigma_{23}+\sigma_{23}\sigma_{13}+\sigma_{12}\sigma_{13}}{\sigma_{12}+\sigma_{23}},
\ee
which is equivalent to (\ref{Triangle}).


\subsection{Edges without current}
In the four-node circuit of Fig. \ref{Fig:Triangle}, the current between nodes 2 and 3 is zero. Thus this circuit is electrically equivalent to the previous one. In particular, the equivalent conductance should not depend on $\sigma_{23}$. It is insightful, however, to see explicitly how $\sigma_{23}$ drops out of the final expression. The $\Sigma$-matrix reads
\be
\Sigma=
\begin{pmatrix}
c_1 & -\sigma_{12} & 0 & -\sigma_{14} \\
-\sigma_{12} & c_{2} & -\sigma_{23} & -\sigma_{24}  \\
0 & -\sigma_{23} & c_{3} & 0  \\
-\sigma_{14}&-\sigma_{24} & 0 & c_{4} . 
\end{pmatrix}
\ee
Here we have already  set the symmetric entries equal. Note that $c_3=\sigma_{23}$. When computing the relevant determinants, we can replace the second row with the sum of rows 2 and 3, and then do the same with columns 2 and 3. After these manipulations, the numerator of (\ref{Ceq}) reads
\be
\det \Sigma^\prime = \begin{vmatrix}
c_1 & -\sigma_{12} & 0 \\
-\sigma_{12} & \sigma_{12}+\sigma_{24} & 0 \\
0 & 0 & \sigma_{23}  
\end{vmatrix}= \sigma_{23}\begin{vmatrix}
c_1 & -\sigma_{12} \\
-\sigma_{12} & \sigma_{12}+\sigma_{24}
\end{vmatrix},
\ee
and $\sigma_{23}$ can be factored out. Repeating the same steps for $\det \Sigma^{\prime\prime}$, we can see that the denominator of (\ref{Ceq}) is also proportional to $\sigma_{23}$. Therefore the equivalent conductance does not depend on this edge conductance. Clearly, the same argument works for any circuit containing an edge that only has one connected node. 
\subsection{The classic resistor cube problem}
The edges of the cube in Fig. \ref{Fig:Cube} are made of 1-$\Omega$ resistors ($\sigma = 1 \,\Omega^{-1}$). We need to find the equivalent resistance of this circuit if the battery terminals are connected across the main diagonal (nodes 1 and 8).
\begin{figure}[h!] 
\centerline{\includegraphics[width=6cm, keepaspectratio]{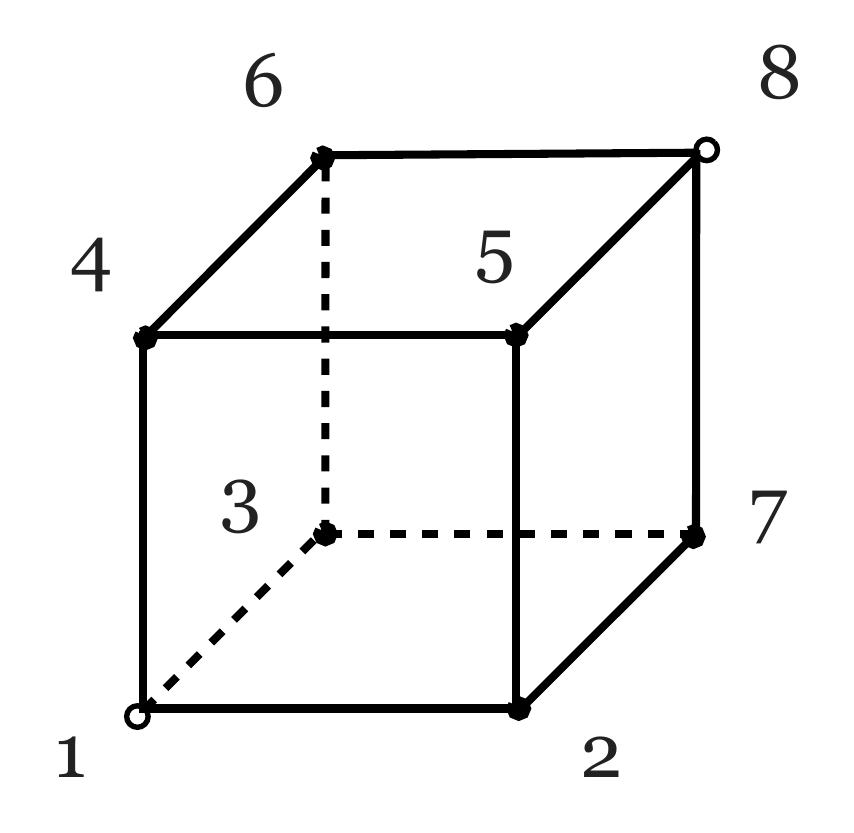}} 
\caption{The classic resistor cube problem. Nodes 1 and 8 are connected to the battery.\label{Fig:Cube}} 
\end{figure}

The $\Sigma$-matrix for this circuit is 
\be\label{Cube}
\Sigma=\begin{pmatrix}
3 & -1 & -1 & -1 & 0 & 0 & 0 & 0\\
-1 & 3 & 0 & 0 & -1 & 0 & -1 & 0\\
-1 & 0 & 3 & 0 & 0 & -1 & -1 & 0\\
-1 & 0 & 0 & 3 & -1 & -1 & 0 & 0\\
0 & -1 & 0 & -1 & 3 & 0 & 0 & -1\\
0 & 0 & -1 & -1 & 0 & 3 & 0 & -1\\
0 & -1 & -1 & 0 & 0 & 0 & 3& -1\\
0 & 0 & 0 & 0 & -1& -1 & -1 & 3\\
\end{pmatrix}
\ee
Computing the relevant determinants, we obtain the equivalent conductance
\be
\sigma_{\rm eq} = \frac{\det \Sigma^\prime}{\det \Sigma^{\prime\prime}}=\frac{384}{320}=\frac{6}{5} \,\,(\Omega^{-1}),
\ee
or $R_{\rm eq}=5/6 \,\,\Omega$, which is the well-known result. Note that while symmetry is crucial for the classical solution of this problem, Eq. (\ref{Ceq}) would still apply for arbitrary edge conductances of the cube.
\subsection{Complete graph}
Consider a circuit with $N$ nodes, such that every node is connected to every other node by an edge of fixed resistance, say 1 $\Omega$ ($\sigma_{ij}=1 \Omega^{-1}$). One method of calculating the equivalent resistance of this circuit is presented in; \cite{Wu} and the answer is $R_{\rm eq}= 2/N$ ($\Omega$). Before applying (\ref{Ceq}), we note that there is a simple physical solution. Assuming that the battery is connected to nodes $1$ and $N$, all the other $(N-2)$ nodes are equivalent, hence their nodal potentials are equal: $V_2=V_3=...=V_{N-1}$. We can thus merge the latter nodes into one ($M$) and obtain a very simple circuit: $(N-2)$ parallel edges from $1$ to $M$, $(N-2)$ parallel edges from $M$ to $N$ and one direct edge from $1$ to $N$. It then straightforward to obtain the above expression for $R_{\rm eq}$.

Let us now consider the relevant determinants of the $\Sigma$-matrix.

\be\label{AllCon}
\det \Sigma^\prime=\begin{vmatrix}
N-1 & -1 & -1&\dots & -1 \\
-1 & N-1 & -1 & \dots & -1\\
-1 & -1 & N-1 & \dots & -1\\
\vdots & \vdots & \vdots & \ddots& \vdots\\
-1 & -1 & -1 & \dots & N-1\\
\end{vmatrix},
\ee
with $(N-1)$ elements in each row and column. Expanding this determinant by the elements of first row, we obtain one term proportional to $\det \Sigma^{\prime\prime}$ and $(N-2)$ other terms, which can be shown to be identical (denote  $\Delta$). In other words,
\be\label{AllCon1}
\det\Sigma^{\prime} =(N-1) \det \Sigma^{\prime\prime}+(N-2) \Delta,
\ee
where
\be
\det\Sigma^{\prime\prime}=\begin{vmatrix}
N-1 & -1 & \dots & -1 \\
-1 & N-1  & \dots & -1\\
\vdots  & \vdots & \ddots& \vdots\\
-1 & -1  & \dots & N-1\\
\end{vmatrix},\quad
\Delta=\begin{vmatrix}
-1 & -1 & \dots & -1 \\
-1 & N-1  & \dots & -1\\
\vdots  & \vdots & \ddots& \vdots\\
-1 & -1  & \dots & N-1\\
\end{vmatrix},
\ee
with $(N-2)$ elements in each row and column. At the same time, the determinant in Eq.(\ref{AllCon}) can be computed by subtracting the second row from the first one
\be\label{AllCon2}
\det \Sigma^\prime=\begin{vmatrix}
N & -N & 0 &\dots & 0 \\
-1 & N-1 & -1 & \dots & -1\\
-1 & -1 & N-1 & \dots & -1\\
\vdots & \vdots & \vdots & \ddots& \vdots\\
-1 & -1 & -1 & \dots & N-1\\
\end{vmatrix}=N(\det\Sigma^{\prime\prime}+\Delta),
\ee
where we expanded the determinant by the elements of the first row. Setting the righthand sides of Eqs. (\ref{AllCon1}) and (\ref{AllCon2}) equal, we can express $\Delta$ as 
\be
\Delta = - \frac{1}{2}\det\Sigma^{\prime\prime}.
\ee
Substituting this result back into (\ref{AllCon2}) yields $\det\Sigma^{\prime}=\frac{N}{2}\det\Sigma^{\prime\prime}$. Finally, applying (\ref{Ceq}), we obtain
\be
\sigma_{\rm eq} = \frac{\det\Sigma^\prime}{\det\Sigma^{\prime\prime}} = \frac{N}{2} \,\,(\Omega^{-1}),
\ee
which agrees with the equivalent resistance $R_{\rm eq}= 2/N$ ($\Omega$).
\section{Discussion}\label{Sec:Discussion}
We have derived a closed formula for the equivalent conductance of an arbitrary circuit. All one needs to know is the edge conductances $\sigma_{ij}$ which give rise to the $\Sigma$-matrix (weighted Laplacian matrix) defined in (\ref{SigmaMatrix}). The equivalent conductance can then be computed as a ratio of two sub-determinants of $\Sigma$ via Eq. (\ref{Ceq}). Interestingly, both $\det \Sigma^\prime$ and $\det \Sigma^{\prime\prime}$ are {\em linear} functions of each individual edge resistance (see Appendix \ref{SigmaFun}).

One important feature of  (\ref{Ceq}) is that it respects the permutation symmetry. Indeed, relabeling any nodes other  than the two connected to the battery terminals (1 and $n$) must not affect the equivalent conductance of the circuit. For example, switching labels $i$ and $j$ would merely result in a minus sign in front of $\det \Sigma^\prime$ and $\det \Sigma^{\prime\prime}$ thus keeping the answer intact.

On physical grounds, the equivalent conductance of an arbitrary circuit (with $\sigma_{ij}<\infty$) can be zero (for a disconnected circuit), but never infinite. One the other hand, the denominator in Eq. (\ref{Ceq}) can vanish if sufficiently many $\sigma_{ij}$ equal zero. This implies that if $\det \Sigma^{\prime\prime}=0$, the other determinant, $\det \Sigma^{\prime}$, must vanish as well. Physically, this is the case when the circuit contains isolated clusters of nodes. Removing such clusters (and reducing the size of the $\Sigma$-matrix) renders the determinants non-zero.

In this paper we focused only on resistive circuits with a single battery. However, the same analysis can be applied to circuits with multiple batteries. This can be done by simply incorporating the additional EMF's into the nodal potential differences. Clearly, ideal batteries would not affect the equivalent resistance between fixed nodes. 

Generalization to capacitive circuits is also straightforward, as equivalent capacitance obeys the same rules (\ref{sigmaPar}) and (\ref{sigmaSer}) as conductance. Thus the final formula (\ref{Ceq}) can be understood in terms of capacitance as well.

In addition, Eq.(\ref{Ceq}) will work for the equivalent impedance (admittance) of an AC-circuit. The only difference would be that admittance is a complex number and either $\det \Sigma^{\prime}$ or $\det \Sigma^{\prime\prime}$ can be zero even for a connected AC-circuit. In fact, setting the determinants to zero one can determine the resonance frequencies of the circuit.

Interestingly, Eq. (\ref{Ceq}) unveils a curious interplay between electrical circuits, matrix algebra, graph theory and its applications to computer science. Specifically, there is a straightforward correspondence between electrical circuits and random walks on graphs, \cite{Doyle} including the concept of {\em escape probability}, which is a direct analog of equivalent resistance. In addition, Eq. (\ref{Ceq}) can help to investigate the connectivity of generic graphs, which is done in, e.g., \cite{Mohar}  using spectral analysis. These interdisciplinary connections are particularly useful, as there is much physical intuition about electrical circuits that could give rise to some less obvious mathematical statements. 

To conclude we note the following. While the method of nodal potentials should probably have place in an E\&M curriculum, \cite{Potential} I would rather not recommend handing formula (\ref{Ceq}) out to students; that is unless the students derive it themselves! Giving the answer away could ruin the excitement of puzzle-solving that students usually associate with circuit analysis. On the other hand, this closed expression can help {\em physics instructors} when creating multiple versions of circuits problems. For they would easily be able to randomize the edge resistances for many students without any worries about checking their students' answers.

\section*{Acknowledgements}
The author would like to thank Dan Boykis, Ken Johnson, Greg Kagan, Patrick Moylan, Stan Ritvin and Artur Tsobanjan for helpful comments and discussions and Bojan Mohar for pointing out the existing graph theoretical results. The author is also grateful to the editors and anonymous reviewers for the suggestions that led to substantial improvements of the paper.

\begin{appendix}
\section{Properties of the determinants}\label{App:Determinant}
\subsection{Positivity of $\det \Sigma^\prime$ and $\det \Sigma^{\prime\prime}$}
In general, determinants contain both positive and negative monomials. However, in Sec. \ref{Sec:Wheatstone} we saw that, after some cancellations, the determinants in the numerator and denominator of (\ref{CeqWheatstone}) had only positive terms. In this appendix, we shall prove that the determinant of $\Sigma^\prime$ is always of this form, that is 
\begin{lemma}\label{lemma3}
Each term in $\det \Sigma^{\prime}$ enters with a plus.
\end{lemma}
\begin{proof}
We shall proceed with the proof by induction in $n$. For example, for $n=3$,
\be
\Sigma^{\prime}=
\begin{pmatrix}
 c_{1} & -\sigma_{12} \\
 -\sigma_{21} & c_{2} 
\end{pmatrix},
\ee
where $c_1=\sigma_{12}+\sigma_{13}$, $c_2=\sigma_{21}+\sigma_{23}$ and $\sigma_{12}=\sigma_{21}$. Thus the determinant
\be
\det \Sigma^{\prime}=c_1 c_2 - \sigma_{12}^2=\sigma_{12}\sigma_{23}+\sigma_{12}\sigma_{13}+\sigma_{13}\sigma_{23}.
\ee
In addition, in Sec. \ref{Sec:Wheatstone} we saw the statement of the Lemma to be true for $n=4$.

Assume now that the statement holds for $n=k$. In order to prove that it also is true for $n=k+1$ ($\Sigma^\prime$ is then $k\times k$), it is sufficient to demonstrate that the coefficient in front of each $\sigma_{ij}$, as it enters $\det \Sigma^\prime$, is positive. Without loss of generality, we can focus on $\sigma_{12}\equiv \sigma_{21}$. In the determinant 
\be
\det \Sigma^\prime=
\begin{vmatrix}
c_1 & -\sigma_{12} & -\sigma_{13} & \dots &-\sigma_{1,k}\\
-\sigma_{21} & c_{2} & -\sigma_{23} & \dots&-\sigma_{2,k}\\
-\sigma_{31} & -\sigma_{32} & c_{3} & \dots &-\sigma_{3,k}\\
\vdots & \vdots & \vdots &\ddots  &\vdots\\
-\sigma_{k,1} & -\sigma_{k,2} & -\sigma_{k,3} & \dots &c_{k}
\end{vmatrix},
\ee
$\sigma_{12}$ appears in the four upper-left entries, including inside $c_1$ and $c_2$, but nowhere else in the rest of the matrix. Using this fact and properties of determinants, we can rewrite the expression so that $\sigma_{12}$ will only appear in one place. Specifically we can replace the second row with the sum of itself and the first one and then repeat this procedure with the same columns 
\begin{eqnarray}
\det \Sigma^\prime &=&
\begin{vmatrix}
c_1 & -\sigma_{12} & -\sigma_{13} & \dots &-\sigma_{1,k} \\
c^\prime_1 & c^\prime_{2} & -(\sigma_{13}+\sigma_{23}) & \dots&-(\sigma_{1,k}+\sigma_{2,k})\\
-\sigma_{31} & -\sigma_{32} & c_{3} & \dots &-\sigma_{3,k}\\
\vdots & \vdots & \vdots &\ddots &\vdots \\
-\sigma_{k,1} & -\sigma_{k,2} & -\sigma_{k,3} & \dots &c_{k}
\end{vmatrix} \\
&=&
\begin{vmatrix}
c_1 &  c^\prime_{1} &-\sigma_{13} & \dots &-\sigma_{1,k} \\
c^\prime_1 & c^\prime_1+c^\prime_{2} & -(\sigma_{13}+\sigma_{23}) & \dots&-(\sigma_{1,k}+\sigma_{2,k})\\
-\sigma_{31} & -(\sigma_{31}+\sigma_{32}) & c_{3} & \dots &-\sigma_{3,k}\\
\vdots & \vdots & \vdots &\ddots &\vdots \\
-\sigma_{k,1} & -(\sigma_{k,1} +\sigma_{k,2}) & -\sigma_{k,3} & \dots &c_{k}
\end{vmatrix} .\nonumber
\end{eqnarray} 
Here $c^\prime_1=c_1-\sigma_{12}$ and $c^\prime_2=c_2-\sigma_{12}$ do not contain $\sigma_{12}$. Hence the only entry that depends on $\sigma_{12}$ is $c_1$, and it does so linearly. We can expand the latter determinant as
\begin{eqnarray}
\det \Sigma^\prime &=&\sigma_{12}\begin{vmatrix}
 c^\prime_{2} & -(\sigma_{13}+\sigma_{23}) & \dots&-(\sigma_{1,k}+\sigma_{2,k})\\
 -(\sigma_{31}+\sigma_{32}) & c_{3} & \dots &-\sigma_{3,k}\\
\vdots & \vdots &\ddots &\vdots \\
 -(\sigma_{k,1} +\sigma_{k,2}) & -\sigma_{k,3} & \dots &c_{k}
\end{vmatrix}\nonumber\\
&+&(\text {terms not containing } \sigma_{12}).\label{DetExp}
\end{eqnarray}
Define $\sigma^\prime_{i,2}:=\sigma_{i,1}+\sigma_{i,2}$ for $i=3,...,k$. Then the determinant multiplying $\sigma_{12}$ in Eq. (\ref{DetExp}) will be of the same form as the original determinant of $\Sigma^\prime$. The size of this determinant is $(k-1)\times (k-1)$ (which corresponds to $n=k$), so by the induction hypothesis it must be positive. Since the choice of $\sigma_{12}$ was arbitrary we have proved that the coefficient in front of each edge conductance in $\Sigma^\prime$ is positive. Therefore $\Sigma^\prime$, as a polynomial in $\sigma$'s, has only positive terms, which proves the lemma.
\end{proof}
Moreover, as $ \Sigma^{\prime\prime}$ has a form very similar to that of $\Sigma^\prime$, the proof above would work for its determinant as well. 

\subsection{Equivalent conductance as a function of edge conductances}\label{SigmaFun}
From the consideration above it follows that both $\det \Sigma^{\prime}$ and $\det \Sigma^{\prime\prime}$ are linear functions of edge conductances. Thus the equivalent conductance can be written as
\be\label{sigmaStar}
\sigma_{\rm eq}=\frac{A\sigma_*+B}{C\sigma_*+D}
\ee
for an arbitrary edge conductance $\sigma_*$. Here $A$, $B$, $C$, and $D$ depend on all the edge conductances other than $\sigma_*$.

There is an important property of the equivalent conductance that is also worth mentioning. In Sec. \ref{Sec:Wheatstone} we saw that a special arrangement of some edge conductances (1, 2, 4 and 5) resulted in zero current through the middle edge (3). In that case, it is easy to see that $\sigma_{\rm eq}$ does not depend on $\sigma_3$. Once such a symmetry is recognized, one can do two things without affecting the equivalent conductance:   
\begin{itemize}
\item{Throw the middle edge away, i.e., set $\sigma_3=0$. This can be done, since there is no current through this edge.}
\item{Short-circuit the top and bottom nodes, i.e., put $\sigma_3 \rightarrow \infty$. This can be done, since the nodal potentials $V_2$ and $V_3$ are equal.}
\end{itemize}
In both cases, the resulting circuit can be easily simplified and the equivalent conductance can be computed according to Eqs. (\ref{sigmaPar}) and (\ref{sigmaSer}). Importantly, these two resulting circuits are {\em different}, but have the same equivalent conductance. This may not be as obvious for a more complicated circuit. Suppose there is special edge (with conductance $\sigma_*$) such that performing the two operations above yields the same equivalent conductance. We can prove then that $\sigma_{\rm eq}$ does not depend on $\sigma_*$, as follows. Requiring in (\ref{sigmaStar}) that $\sigma_{\rm eq}(0)=\sigma_{\rm eq}(\infty)$ yields $A/C=B/D$, which implies that $\sigma_*$ drops out from the equivalent conductance.

\section{Simplifiable circuits}\label{App:Simplifiable}
In this section we consider circuits that contain elements in series or in parallel, as well as circuits that can be reduced for symmetry reasons. Specifically, we are interested in the form of the conductance matrix $\sigma_{ij}$ and the corresponding $\Sigma$-matrix.  

\subsection{Connection in series}\label{App:Series}
If there is a pair of edges in series, the node ($k$) shared by these edges would be two-valent and the corresponding row/column in $\sigma_{ij}$ would have only two non-zero entries, say, $\sigma_{kl}$ and $\sigma_{km}$. This also implies that the nodes $l$ and $m$ are not connected directly, i.e., $\sigma_{lm}=0$. From physics we know that the two edges can be replaced by one with the equivalent conductance given by (\ref{sigmaSer}). Therefore, we can reduce the size of the conductance matrix by crossing out the $k^{\rm th}$ row and column and by setting $\sigma_{lm}\equiv \sigma_{ml}:=\left(\sigma_{kl}^{-1}+\sigma_{km}^{-1}\right)^{-1}$.

It is also insightful to investigate this statement mathematically. Without loss of generality, we can set $k=2$, $l=3$, and $m=4$. Then the (top-left part of the) $\Sigma^\prime$-matrix for such a circuit would look like
\be\label{SigmaSer}
\Sigma^\prime=
\begin{pmatrix}
c_1 & 0 & -\sigma_{13} &-\sigma_{14} &-\sigma_{15} & \dots \\
0 & c_{2} & -\sigma_{23} &-\sigma_{24} & 0 & \dots\\
-\sigma_{31} & -\sigma_{32} & c_{3} & 0& -\sigma_{35} &\dots \\
-\sigma_{41} & -\sigma_{42} & 0  &c_{4} &-\sigma_{45} &\dots \\ 
-\sigma_{51} & 0 & -\sigma_{52}  &-\sigma_{53} &c_{5} &\dots \\ 
\vdots & \vdots & \vdots   &\vdots & \vdots &\ddots
\end{pmatrix}.
\ee
As explained above, $\sigma_{34}=\sigma_{43}=0$. Since node 2 is only connected to 3 and 4, the non-zero entries of the second row and column are $\sigma_{23}=\sigma_{32}$, $\sigma_{24}=\sigma_{42}$ and $c_2=\sigma_{23}+\sigma_{24}$. On the other hand, the matrix describing the reduced circuit is 
\be\label{SigmaTildeSer}
\tilde\Sigma^\prime=
\begin{pmatrix}
c_1 &  -\sigma_{13} &-\sigma_{14} &-\sigma_{15} & \dots \\
-\sigma_{31} & \tilde c_{3} & \tilde\sigma_{34}& -\sigma_{35} &\dots \\
-\sigma_{41} & \tilde\sigma_{43}  &\tilde c_{4} &-\sigma_{45} &\dots \\ 
-\sigma_{51} & -\sigma_{53}  &-\sigma_{54} &c_{5} &\dots \\ 
\vdots & \vdots & \vdots   &\vdots &\ddots
\end{pmatrix}.
\ee
Here $\tilde\sigma_{34}=\left(\sigma_{23}^{-1}+\sigma_{24}^{-1}\right)^{-1}\equiv \sigma_{23}\sigma_{24}/(\sigma_{23}+\sigma_{24})$, whereas $\tilde c_3$ and $\tilde c_4$ include $\tilde\sigma_{34}=\tilde\sigma_{43}$. We can similarly introduce the sub-matrices $\Sigma^{\prime\prime}$ and $\tilde\Sigma^{\prime\prime}$. We shall now prove the following 
\begin{lemma}\label{lemma4}
\be\label{CeqSer}
\sigma_{\rm eq} = \f{\det \Sigma^{\prime}}{\det \Sigma^{\prime\prime}} = \f{\det \tilde\Sigma^{\prime}}{\det \tilde\Sigma^{\prime\prime}}=\tilde\sigma_{\rm eq}
\ee
\end{lemma}
\begin{proof}
In what follows, for the sake of compactness we omit the dots in the determinants. We start by expanding $\det \Sigma^{\prime}$ from (\ref{SigmaSer}) in the elements of the second row
\bq
\det \Sigma^\prime &=&
c_2\begin{vmatrix}
c_1 &  -\sigma_{13} &-\sigma_{14} &-\sigma_{15}  \\
-\sigma_{31} & c_{3} & 0& -\sigma_{35} \\
-\sigma_{41} & 0 &c_{4} &-\sigma_{45}\\ 
-\sigma_{51} & -\sigma_{53}  &-\sigma_{54} &c_{5} 
\end{vmatrix}\\
&+&\sigma_{23}\begin{vmatrix}
c_1 &  0 &-\sigma_{14} &-\sigma_{15}  \\
-\sigma_{31} & -\sigma_{32} & 0& -\sigma_{35} \\
-\sigma_{41} & -\sigma_{42}  &c_{4} &-\sigma_{45}\\ 
-\sigma_{51} & 0  &-\sigma_{54} &c_{5} 
\end{vmatrix}+
\sigma_{24}\begin{vmatrix}
c_1 &  -\sigma_{13} &0&-\sigma_{15}  \\
-\sigma_{31} & c_{3} & -\sigma_{32}& -\sigma_{35} \\
-\sigma_{41} & 0  &c_{4} &-\sigma_{42}\\ 
-\sigma_{51} & -\sigma_{53}  &0 &c_{5} 
\end{vmatrix}. \nonumber
\eq
Note that we switched columns 2 and 3 in the last determinant. Using that $c_2=\sigma_{23}+\sigma_{34}$, we can collect similar terms as 
\be
\det \Sigma^\prime =
\sigma_{23}\begin{vmatrix}
c_1 &  -\sigma_{13} &-\sigma_{14} &-\sigma_{15}  \\
-\sigma_{31} & c_{3}-\sigma_{32} & 0& -\sigma_{35} \\
-\sigma_{41} & -\sigma_{42} &c_{4} &-\sigma_{45}\\ 
-\sigma_{51} & -\sigma_{53}  &-\sigma_{54} &c_{5} 
\end{vmatrix}\\
+\sigma_{24}\begin{vmatrix}
c_1 &  -\sigma_{13} &-\sigma_{14} &-\sigma_{15}  \\
-\sigma_{31} &  c_3& -\sigma_{32}& -\sigma_{35} \\
-\sigma_{41} &   0&c_{4}-\sigma_{42} &-\sigma_{45}\\ 
-\sigma_{51} & -\sigma_{53}  &-\sigma_{54} &c_{5} 
\end{vmatrix}.
\ee
It is now easy to see that the sum of columns 2 and 3 is the same for both determinants. Thus we can replace column 2 in each term with this sum, which would allow us to combine the two determinants into one
\be
\begin{vmatrix}
c_1 &  -(\sigma_{13}+\sigma_{14}) &-(\sigma_{23}+\sigma_{24})\sigma_{14} &-\sigma_{15}  \\
-\sigma_{31} & c_{3}-\sigma_{23} & \sigma_{23}\sigma_{24}& -\sigma_{35} \\
-\sigma_{41} & c_4-\sigma_{42} &(\sigma_{23}+\sigma_{24})c_{4}-\sigma_{24}^2 &-\sigma_{45}\\ 
-\sigma_{51} & -(\sigma_{53}+\sigma_{54})  &-(\sigma_{23}+\sigma_{24})\sigma_{54} &c_{5} 
\end{vmatrix}.
\ee
We can now factor $(\sigma_{23}+\sigma_{24})$ out of the third column and use $(\sigma_{23}\sigma_{24})/(\sigma_{23}+\sigma_{24})\equiv \tilde\sigma_{34}$ (in the second entry) together with $\sigma^2_{24}/(\sigma_{23}+\sigma_{24})\equiv \sigma_{24}-\tilde\sigma_{34}$ (in the third entry) to obtain
\be
(\sigma_{23}+\sigma_{24})\begin{vmatrix}
c_1 &  -(\sigma_{13}+\sigma_{14}) &-\sigma_{14} &-\sigma_{15}  \\
-\sigma_{31} & c_{3}-\sigma_{23} & \tilde\sigma_{34}& -\sigma_{35} \\
-\sigma_{41} & c_4-\sigma_{42} &c_{4}-\sigma_{24}+\tilde\sigma_{34} &-\sigma_{45}\\ 
-\sigma_{51} & -(\sigma_{53}+\sigma_{54})  &-\sigma_{54} &c_{5} 
\end{vmatrix},
\ee
where the third diagonal entry is precisely $c_{4}-\sigma_{24}+\tilde\sigma_{34}\equiv \tilde c_4$. Similarly $c_{3}-\sigma_{23}+\tilde\sigma_{34}\equiv \tilde c_3$. With that in mind, subtracting the third column from the second one yields
\be
\det \Sigma^\prime =(\sigma_{23}+\sigma_{24})\begin{vmatrix}
c_1 &  -\sigma_{13}&-\sigma_{14} &-\sigma_{15}  \\
-\sigma_{31} & \tilde c_{3} & -\tilde\sigma_{34}& -\sigma_{35} \\
-\sigma_{41} & -\tilde \sigma_{34}&\tilde c_{4} &-\sigma_{45}\\ 
-\sigma_{51} & -\sigma_{53}  &-\sigma_{54} &c_{5} 
\end{vmatrix} \equiv (\sigma_{23}+\sigma_{24})\det \tilde\Sigma^\prime.
\ee
Repeating the same consideration without the first row and column we would get 
\be
\det \tilde\Sigma^{\prime\prime} =(\sigma_{23}+\sigma_{24})
\begin{vmatrix}
 \tilde c_{3} & -\tilde\sigma_{34}& -\sigma_{35} \\
 -\tilde \sigma_{34}&\tilde c_{4} &-\sigma_{45}\\ 
 -\sigma_{53}  &-\sigma_{54} &c_{5} 
\end{vmatrix}\equiv (\sigma_{23}+\sigma_{24})\det \tilde\Sigma^{\prime\prime},
\ee
from which (\ref{CeqSer}) follows.
\end{proof}

\subsection{Connection in parallel}
If two nodes are connected by multiple edges, those edges can be replaced by one with the total conductance from (\ref{sigmaPar}). This is the reason we did not consider multiple edges from the start, without any loss of generality. 

\subsection{Short-circuiting two nodes}
The two above simplifications do not affect the equivalent conductance of the whole circuit. A more interesting scenario (that does, in general) is short-circuiting two nodes, say $i$ and $j$, by connecting them with an ideal wire. Mathematically this corresponds to letting $\sigma_{ij}\rightarrow \infty$, whereas the physical implication is the equality of the corresponding nodal potentials $V_i=V_j$. Such nodes can be simply merged together, becoming one. After the merging, some edges that were not parallel may become parallel. Thus we can use the idea of the previous paragraph.

In essence, this short-circuiting eliminates one unknown potential, hence reduces the size of the conductance matrix. Instead of the two rows/columns ($i$ and $j$) we have only one. The entries of this new row/column are given by 
\begin{equation}\label{shortC}
\sigma^\prime_{ik}=\sigma_{ik}+\sigma_{jk},
\end{equation}
 for any $k\neq i,j$. 

So far we have not discussed the diagonal elements of $\sigma_{ij}$. One reason being is that they do not affect the $\Sigma$-matrix, which has been most relevant in this paper. In fact, we can generalize the procedure of going from $\sigma_{ij}$ to $\Sigma$ as follows
\begin{equation}\label{sigma_to_Sigma}
\Sigma_{ij}={\rm Diag}(c_1, c_2, ..., c_n) - \sigma_{ij},
\end{equation}
with the same $c$'s as before: $c_{i}=\sum\limits_{j=1}^n \sigma_{ij}$. Clearly, no matter what $\sigma_{ii}$ are, they drop out from $\Sigma$, according to (\ref{sigma_to_Sigma}). 

At the same time, what would be the meaning of $\sigma_{ii}$? How can a node be connected to itself? Now, in the spirit of the short-circuiting scenario above, one can think of each node as ``self short-circuited''. In other words, $\sigma_{ii}=\infty$. In fact, this observation will make the short-circuiting recipe (\ref{shortC}) valid for diagonal elements as well. Notice that all of the manipulations with $\sigma_{ij}$ discussed in this appendix work in the same exact way for the $\Sigma$-matrix. In other words, these manipulations ``commute'' with (\ref{sigma_to_Sigma}).

To conclude, we would like to comment of the symmetry issue. Some nodal potentials may turn out to be equal on symmetry grounds, even if the corresponding nodes are not short-circuited. For instance, as we pointed out in Section \ref{Sec:Discussion}, the Wheatstone bridge circuit has $V_2=V_3$, if (\ref{WheatSym}) is satisfied. Another example is the classic resistor cube problem, where there are triples of (not connected) nodes having the same potential. Importantly, in these symmetric situations, the short-circuiting does not affect the equivalent conductance. As such symmetries are not as manifest in more complicated circuits, it would be interesting to come up with a way of detecting the equipotential nodes by looking at the form of the $\Sigma$ matrix (\ref{Cube}).
\end{appendix}


\begin{thebibliography}{10}
\bibitem{Topology}
Here we assume that all the relevant nodes are electrically connected to the battery and that no circuit node is connected to the baterry terminal via an ideal wire, i.e., via an edge of zero resistance.
\bibitem{Connectivity}
It could be the case that all conductances but, say, $\sigma_1$ were equal zero. Then on physical grounds, that would fix $V_2=V_1=\E$. Mathematically, it can be seen from (\ref{V2V3}) by giving $\sigma_3$ (which would not affect the circuit connectivity or $V_2$ anyway) a small non-zero value. The latter would result in cancellations and $V_2=\E$.
\bibitem{Potential}
Along the lines of Ref. \onlinecite{Rosenthal}, this could help to reinforce students' understanding and appreciation of the concept of electric potential.
\end{thebibliography}
\end{document}